%% file: 2020-FUN-Hangman-BarbaySubercaseaux.tex
\documentclass[a4paper,english,cleveref,autoref]{lipics-v2019}
\usepackage{mdframed}

\usepackage[linesnumbered,ruled,vlined]{algorithm2e}
\usepackage{amsmath}
\usepackage{amssymb,caption}
\usepackage[hang,flushmargin]{footmisc}
\usepackage{subcaption}

\usepackage[dvipsnames]{xcolor}
\usepackage{tikz}
\usetikzlibrary{arrows}
\SetKwInput{KwInput}{Input}                
\SetKwInput{KwOutput}{Output}  

\usepackage{bm}
\usepackage{stmaryrd}            

\newcommand{\csproblem}[2]{
\begin{center}
\fbox{
	\parbox{0.95\textwidth}{
	\vspace*{0.2em}
	\begin{center}
		\begin{minipage}{0.9\textwidth}
		\begin{flushright}
			\textsc{#1}
		\end{flushright}
\centering\setlength\parskip\medskipamount   
	#2
      \end{minipage}
	\end{center}\vspace*{0.6em}
	}
}\end{center}}
\newcommand{\greedycheater}{\texttt{GreedyCheater}}
\newcommand{\nphard}{$\mathrm{NP}$-hard}
\newcommand{\nphardness}{$\mathrm{NP}$-hardness}
\newcommand{\MDS}{$\textsc{Minimum Dominating Vertex Set}$}
\newcommand{\conphard}{$\mathrm{coNP}$-hard}
\newcommand{\conphardness}{$\mathrm{coNP}$-hardness}
\newcommand{\filter}{\mathcal{F}}

\title{The Computational Complexity of Evil Hangman}
\titlerunning{The Computational Complexity of Evil Hangman}

\author{J\'er\'emy Barbay}
{Department of Computer Science, University of Chile \and \url{http://barbay.cl}}%
{jeremy@barbay.cl}%
{https://orcid.org/0000-0002-3392-8353}%
{}

\author{Bernardo Subercaseaux}
{Department of Computer Science, University of Chile \and Millennium Institute for Foundational Research on Data (IMFD), Chile}%
{bsuberca@dcc.uchile.cl}%
{https://orcid.org/0000-0003-2295-1299}
{}

\authorrunning{J. Barbay and B. Subercaseaux}

\Copyright{J\'er\'emy Barbay and Bernardo Subercaseaux}



\relatedversion{}

\acknowledgements{We want to thank Robinson Castro, who first introduced us to this problem. We also thank Nicol\'{a}s Sanhueza-Matamala, Alex Meiburg and the anonymous reviewersx for their insightful comments and discussion.}

\nolinenumbers 

\EventEditors{Martin Farach-Colton, Giuseppe Prencipe, and Ryuhei Uehara}
\EventNoEds{3}
\EventLongTitle{10th International Conference on Fun with Algorithms (FUN 2020)}
\EventShortTitle{FUN 2020}
\EventAcronym{FUN}
\EventYear{2020}
\EventDate{September 28--30, 2020}
\EventLocation{Favignana Island, Sicily, Italy}
\EventLogo{}
\SeriesVolume{157}
\ArticleNo{23}

\begin{document}

\maketitle

\begin{abstract} 
The game of Hangman is a classical asymmetric two player game in which one player, the setter, chooses a secret word from a language, that the other player, the guesser, tries to discover through single letter matching queries, answered by all occurrences of this letter if any. In the Evil Hangman variant, the setter can change the secret word during the game, as long as the new choice is consistent with the information already given to the guesser. We show that a greedy strategy for Evil Hangman can perform arbitrarily far from optimal, and most importantly, that playing optimally as an Evil Hangman setter is computationally difficult. The latter result holds even assuming perfect knowledge of the language, for several classes of languages, ranging from Finite to Turing Computable. The proofs are based on reductions to Dominating Set on 3-regular graphs and to the Membership problem, combinatorial problems already known to be computationally hard.
\end{abstract}

\input{hangman-camera-ready.tex}

\bibliography{biblio}

\end{document}

%% file: hangman-camera-ready.tex
\section{Introduction}
\label{sec:intro}

\newcommand{\blank}{\bot}

The \textsc{Hangman}'s game is a classical asymmetric two player game, where one player, denoted as \textit{the setter} keeps a secret word $w$, that the other player, denoted as \textit{the guesser}, needs to guess. The game starts with both players agreeing on a maximum number of guesses $d$, and the setter communicating to the guesser an integer $k$, the length of the word $w$. Then, in every turn the guesser makes a query with a single letter $s$, and the setter reveals every appearance of $s$ in $w$, if any.
A query $s$ is said to have failed if there are no occurrences of $s$ in $w$. The game ends either when $w$ is fully revealed, in which case we say the guesser won, or when the guesser has made more than $d$ failed queries, in which case we say the setter won the game.


%
The \textsc{Evil Hangman} variant is a little twist to the game that has been widely used as a programming assignment~\cite{Parlante2011}. In this variant, the setter can change the secret word as often as she wants during the game, as long as at the end she's able to reveal a word that is consistent with all the answers given thus far to the guesser's queries. In the programming assignment~\cite{Parlante2011}, students are given the task of implementing what we call the \greedycheater, an evil setter that decides to answer each query with the heuristic of keeping the dictionary of consistent words to reveal at the end as big as possible.


A natural question is whether the algorithm {\greedycheater} is the best one can do when trying to maximize the number of guesses a guesser requires to discover the secret word, and if it is not, then \textbf{what would be the best strategy to maximize the number of questions required to guess a word in a} \textsc{Evil Hangman} \textbf{game}. Such questions can be asked in various contexts, such as when the language from which the words are chosen and guessed is a natural one (i.e. a finite set of words of fixed length $k$), or as when the language is defined more formally (e.g. through a Turing machine, the language being projected to a finite set of words of fixed length $k$).


We formalize such various contexts of the \textsc{Hangman} and \textsc{Evil Hangman} game, and the related computational problems, in Section~\ref{sec:preliminaries}.
As a preliminary result, we show that the number of guesses generated by the {\greedycheater} strategy is not only sub-optimal on some trivial examples, but that it can be arbitrarily worse than the optimal strategy on an infinite family of scenarios (in Section~\ref{sec:greedy}).
Our main result is about finding an optimal strategy for the game of \textsc{Evil Hangman}: while it is clear that there is an exponential time \textit{minimax} strategy for playing optimally as an evil setter (see Equation~\ref{eq:recursive}), we prove (in Section~\ref{sec:hardness}) that such a running time is essentially optimal in the worst case, as the problem of deciding for a given set $L$ of words and a maximal number $d$ of guesses whether there is a strategy forcing more than $d$ failed guesses, is \conphard\ (and hence, deciding if there is a winnning stategy for the guesser is \nphard).
Pushing further such results, we prove (in Section~\ref{sec:turing-recogn-dict}) that playing optimally is $\mathrm{PSPACE}$-complete when the game is played over the language defined by a Context Sensitive Grammar, and undecidable when played over the language defined by a Turing Machine.
We conclude in Section~\ref{sec:openquestions} with a discussion of other minor results and potential open directions for further research. 

Throughout this paper, we consider  the problem in the context of a perfect guesser, with perfect knowledge of the language.

\section{Preliminaries}
\label{sec:preliminaries}

Before stating our results, we describe more formally the original \textsc{Hangman} game (Section \ref{sec:hangman}), the \textsc{Evil Hangman} variant (Section~\ref{sec:evil-hangman}), and a formalization of how to measure the quality of strategies for the \textsc{Evil Hangman} variant (Section~\ref{sec:eval-textsc-hangm}).

\subsection{Hangman}\label{sec:hangman}

The game starts with a word length, $k$, and a parameter $d$ stating the number of failed guesses allowed, being agreed upon between the players.  A game of \textsc{Hangman} is played over an alphabet~$\Sigma$ initially set to $[1..\sigma]$ and a (potentially infinite) language $L$ projected to words of length $k$ on the remaining alphabet by intersecting it with $\Sigma^k$. The alphabet will be progressively reduced during the game to capture the fact that symbols which have already been guessed should not be part of the game anymore. We will use an extra symbol, not present in~$\Sigma$ that denotes a letter not yet revealed. Players often use an underscore $(\_)$ for this, but we will use the symbol $\blank$ instead, for better readability.
	
As the game goes by turns, we can define the state of the game in terms of which letters have been discarded, and what action is taken, on the $i$-th turn (turns are $1$-indexed).  The game starts with the setter revealing $M_0 = \blank^k$, which represents what the guesser knows about the word at that point. Then, on the $i$-th turn, the guesser makes a query with the symbol $s_i$, and the setter replies with the mask $M_i$, which is equal to $M_{i-1}$ except possibly for some occurrences of the $\blank$ symbol, that have been changed to $s_i$. Figure~\ref{fig:exampleTraditional} presents an example of a traditional game of \textsc{Hangman}. 

	\begin{figure}
	\begin{mdframed}[linecolor=black!30, backgroundcolor=black!3]
	\centering
	\begin{subfigure}[b]{0.5\textwidth}
         \centering
				\begin{enumerate}
					\item $w = fun$
					\item $M_0 = \blank \blank \blank$
					\item $\bm{s_1 = e}$
					\item $M_1 = \blank \blank \blank$
					\item $s_2 = n$
					\item $M_2 = \blank \blank n$
					\item \bm{$s_3 = a$}
					\item $M_3 = \blank \blank n$
					\item $s_4 = u$
					\item $M_4 = \blank un$
					\item $s_5 = f$
					\item \bm{$M_5 = fun$}
				\end{enumerate}
         \caption{Traditional \textsc{Hangman} (guesser wins). Note that $s_1$ and $s_3$ are failed guesses.}
        \label{fig:exampleTraditional}
     \end{subfigure}
     \hfill
	 \begin{subfigure}[b]{0.45\textwidth}
         \centering
				
				\begin{enumerate}
					\item $w_0 = run$
					\item $M_0 = \blank \blank \blank$
					\item $s_1 = u$
					\item $w_1 = run \implies M_1 = \blank u \blank$
					\item $s_2 = n$
					\item $w_2 = run \implies M_2 = \blank u n$
					\item \bm{$s_3 = r$}
					\item $w_3 = pun \implies M_3 = \blank u n$
					\item \bm{$s_4 = p$}
					\item $w_4 = sun \implies M_4 = \blank u n$
					\item \bm{$s_5 = s$}
					\item $w_5 = fun \implies M_5 = \blank u n$
				\end{enumerate}
         \caption{\textsc{Evil Hangman} (setter wins). Note that $s_3$, $s_4$ and $s_5$ are failed guesses.}
        \label{fig:exampleEvil}
     \end{subfigure}

	\caption{Example games of \textsc{Hangman} ($k=3, d=3$) over the Latin lowercase alphabet and using English as a language.
}
	\label{fig:example}
	\end{mdframed}
	\end{figure}

We define as well the operations $a \oplus b$ to be the result of replacing every $\blank$ character in $a$ by its corresponding character in $b$, and $(a \cap b)_i$ to be $a_i$ if  $a_i = b_i$, and $\blank$ otherwise.
Now we state that given a secret word $w$, we can compute $M_i$ after a guess $s_i$ as 
$
	M_{i+1} = M_i \oplus (w \cap s_i^k).
$

It is also helpful to define  $B_i$ as the indices in $M_i$
that have the symbol $\blank$, as this is the set that the setter can choose a subset from when answering a query. We use as well the notation $s^B$, with $s$ being a symbol and $B$ a subset of $[1..k]$, where $k$ is implicit, to describe the word $w$ of size $k$ such that $w_i = s$ if $i \in B$ and $w_i = \bot$ otherwise. Finally, for a language $L$ and a mask $M$, we abbreviate the set $\{ w \in L \mid M \preceq w \}$ as $\filter(L, M)$.

\subsection{Evil Hangman}\label{sec:evil-hangman}

In the evil version of the game, the setter can choose to change the secret word, even every turn, as long as the new choice is consistent with the answers given so far. We say that $w_i$ is the secret word before the setter reveals $M_i$. Figure~\ref{fig:exampleEvil} presents an example of a an \textsc{Evil Hangman} game.

 In order to define what the required \textit{consistency} exactly means, we define a relation $\preceq$, such that for two words $a$ and $b$ of length $k$

$$
	a \preceq b \iff (a_i = \blank) \lor (a_i = b_i), \quad \text{for all } i \in \{1, \ldots, k\} 
$$

Intuitively, we say that $b$ is consistent with $a$, as the only differences they can have are positions where $a$ had not been revealed yet. We can now state our consistency restrictions as follows:

	\begin{enumerate}[\hspace{.4\linewidth}1.]
		\item $w_i \in L$ 
		\vspace{0.8em}
		\item $M_i \preceq w_i$ 
		\vspace{0.8em}
		\item $w_{i_j} \not\in \{s_1, ..., s_{i-1} \},  \forall j \in B_i$	
	\end{enumerate}

The first rule requires the partial secret words to be part of the language, the second one requires that partial secret words do not differ with the exposed mask, and the third one requires that the symbols not yet revealed to the guesser do not match any of the previous guesses.
The latter is simply captured by the (dynamic) alphabet~$\Sigma$, which contains only the remaining possible guesses.

\subsection{Evaluation of \textsc{Evil Hangman} strategies} \label{sec:eval-textsc-hangm}

An evil setter strategy, over a language $L$, is a function $A_L$ that takes a mask $M_i$, a guess $s_i$ and returns $A_L(M_i, s_i) = M_{i+1}$ a new mask, such that $M_i \preceq M_{i+1}$, and there exists $w \in L$ such that $M_{i+1} \preceq w$.
We define the function $W$ to measure how good a particular situation is for the setter. 
A situation is a tuple $(M, \Sigma)$ where $M$ is a word mask and~$\Sigma$ is the remaining alphabet (letters that have not been guessed so far).
We define $W(M, \Sigma, A_L)$ for an adversarial setter $A_L$ as the minimum number of failed guesses that any player would have to do in order to reveal a full word, starting from $M$ over the alphabet~$\Sigma$.  We can formalize the function $W$ inductively as follows, where we use the Iverson bracket notation~\cite{IversonBracketNotationWikipedia}, namely, for any predicate $P$, the expression $\llbracket P \rrbracket$ equals $1$ if $P$ is true and $0$ otherwise:

\begin{equation}
	W(M, \Sigma, A_L) = \begin{cases}
		0 & \text{if } M \in L \\
		\min\limits_{s \in \Sigma} \Big\{ W\left(
                  \begin{array}{c}                    
                    A_L(M, s), \\
                    \Sigma \setminus \{ s\}, \\
                    A_{\filter(L, A_L(M, s))}
                  \end{array}
                \right) +  \llbracket A_L(M, s) = M \rrbracket  \Big\}  & \text{otherwise}
	\end{cases}
	\label{eq:inductive}
\end{equation}

We define $OPT$ as an adversary such that $W(M, \Sigma, OPT)$ is maximum for every $M$ and~$\Sigma$.
Noting $B_M$  the set of indices of a word $M$ containing $\bot$ allows us to state $OPT$ explicitly as a recursive formula:


\begin{equation}
	OPT_L(M, s, \Sigma) = \max_{
	\substack{B \subseteq B_M \\ 
	{ \filter(L, M \oplus s^B) \neq \varnothing}}} 
	\left\{ 
		\llbracket B = \emptyset \rrbracket + \min_{s' \in \Sigma}  \left\{ OPT_{\filter(L, M \oplus s^B)} \Big( M \oplus s^B, s', \Sigma \setminus \{ s \} \Big)  \right\}
 	\right\}
	\label{eq:recursive}
\end{equation}

We will refer to this optimal adversary in the next section, when showing that the greedy algorithm's competitive ratio is not bounded by a constant, and thus that such an algorithm can perform arbitrarily bad on an infinite family of instances.
\section{The greedy cheater}
\label{sec:greedy}

Algorithm~\ref{alg:greedy} presents a pseudo code for the algorithm \texttt{GreedyCheater}, a pretty intuitive and efficient algorithm for the setter, that is often given as a programming assignment~\cite{Parlante2011} in colleges across the US. The idea is to answer every query in such a way that the number of words remaining in the dictionary that are consistent with the answer is maximized.

\begin{algorithm}
\caption{Pseudo code for the {\greedycheater} algorithm}\label{alg:greedy}
\DontPrintSemicolon
 \KwInput{a mask $M_{i-1}$, a guess $s_i$ and a dictionary $D$}
 \KwOutput{the mask $M_i$}
 
 $bestMaskSize \leftarrow 0$ \;
 $bestMask \leftarrow NIL$ \;
 \For{$B \in 2^{B_{i-1}}$} {
 	$thisMask \leftarrow M_{i-1} \oplus s_i^B$ \;
 	$thisMaskSize \leftarrow | \filter(D, thisMask)|$ \;
 	\If{$thisMaskSize > bestMaskSize$} {
 		$bestMaskSize \leftarrow thisMaskSize$ \;
 		$bestMask \leftarrow thisMask$
 	}
 }
 \Return{$bestMask$}
\end{algorithm}

Not only is the {\greedycheater} algorithm not optimal, it can be arbitrarily bad, which we formalize in the following theorem, and illustrate with an example in Figure~\ref{fig:adverDict}.
\begin{theorem}\label{thm:greedycheaterIsNotCCompetitive}
%
\greedycheater\ is not $c$-competitive  in terms of $W$.
That is, there are no constants $c > 0$ and $b$ such that $c \cdot W(M, \Sigma, \greedycheater) + b \geq W(M, \Sigma, OPT_L)$ for every possible language $L$ and situation $(M, \Sigma)$.
\end{theorem}

\begin{figure}[t]
	\begin{mdframed}[linecolor=black!30,backgroundcolor=black!3]
	\centering

		\begin{tikzpicture}
		\node[] (D) at (-7, 0) 
		{$
			D \left\{
			\begin{array}{c}
				\texttt{abbc}\\
				\texttt{abcb}\\
				\texttt{abcc}\\
				\texttt{dddd}\\
				\texttt{eeee}
			\end{array}
			\right.
		$};
		
		\node[text width=1.7cm] (gasks) at (-4, 0) {
			Guesser asks for '\texttt{a}'
		};
		
		\node[] (a) at (1, 2)
		{$
			\left.
			\begin{array}{c}
				\texttt{abbc}\\
				\texttt{abcb}\\
				\texttt{abcc}\\
			\end{array}
			\right\}		
		$};
		
		\node[] (b) at (1, -2)
		{$
			\left.
			\begin{array}{c}
				\texttt{dddd}\\
				\texttt{eeee}\\
			\end{array}
			\right\}
		$};
		
		\node[text width=2cm] (texta) at (3,2) {
			The guesser can win with no loss
		};
		
		\node[text width=2cm] (textb) at (3,-2) {
			A failed guess can be forced
		};
		
		\path [->] (D) edge node {} (gasks);
		\path [->, sloped] (gasks) edge node [above] {\greedycheater} node [below] {\texttt{a} $\bot \bot \bot$} (a);
		\path [->, sloped, text width =1.8cm] (gasks) edge node [below] {OPT, takes 1 life} node[above] {$\bot \bot \bot \bot$} (b);
		\end{tikzpicture}
	\caption{Example of an adversarial dictionary $D$  of size $n = 5$ over the alphabet $\{a,b,c,d,e\}$  against the \greedycheater\ algorithm, with word length $k = 4$.}
	\label{fig:adverDict}
	\end{mdframed}
\end{figure}

\begin{proof}
We describe how to build an adversarial dictionary $D$ of size $n = 2m+1$, where
\begin{itemize}
\item $m+1$ words start with the symbol $\alpha$, and have only symbols $\beta$ and $\gamma$ in the other positions (both $\beta$ and $\gamma$ must be present) and
\item the remaining $m$ words are of the form $\eta^k$ for symbols $\eta \not\in \{\alpha, \beta, \gamma \}$.
\end{itemize}

Note that this requires $k \geq 1 + \lg(m+1)$ and $\sigma \geq m + 3$, to ensure that we can actually build the first $m+1$ different words with combinations of $\beta$ and $\gamma$ and the last $m$ words with different symbols. This poses no problem, as such $k$ and $\sigma$ exist for every $m$, and thus we can build bad instances of arbitrary size.

Now, upon the guess $\alpha$, the greedy algorithm will answer $\alpha \blank^{k-1}$ (as $m+1$ words start with it, as opposed to the $m$ that do not have it), and then  after guessing $\beta$ and $\gamma$ the guesser will find the word with loss $0$.

On the other hand, an optimal \textsc{Evil Hangman} algorithm would reply with $\blank^k$ to the original guess, and then on any guess with a symbol in $\Sigma \setminus \{\alpha, \beta, \gamma\}$ except for the last one, it would reply with the same mask taking a life from the guesser. Such a strategy makes the guesser lose on the symbol $\alpha$ and at least $m-1$ other symbols, giving a total loss of $m$. This concludes the proof.
\end{proof}


The {\greedycheater} algorithm  can be efficiently implemented. But given that it can perform arbitrarily far from the optimum, a natural question is whether there is an efficient algorithm that achieves optimality, which we explore in the next section.

\section{Hardness of Finding an Optimal Evil Adversary}
\label{sec:hardness}

Consider the decision problem of whether the setter can \textit{win} the game against any possible guesser. We restrict our analysis to finite languages in this section, and explore generalized languages (languages that are higher up in the Chomsky hierarchy) in Section \ref{sec:turing-recogn-dict}.

Our main computational problem is, given a finite language $L$ of words of length $k$ on an alphabet of size $\sigma$, to decide if an evil Hangman setter has a winning strategy, where winning is defined with respect to the number $d$ of failed guesses that the guesser is allowed:
\csproblem{Evil Hangman}{Given a finite language $L$, where every word has some fixed size $k$, and an integer $d$, decide whether it is possible for a cheating setter to play in such a way that no guesser can get the secret word without making at least $d$ failed guesses.}

We prove the difficulty of \textsc{Evil Hangman} through a reduction to the problem of \MDS\ in $3$-regular graphs, by taking such a graph and encoding it as a language.
Intuitively, we will build an alphabet by associating a different symbol to each vertex of the graph, and the language will be constructed by associating a word to each vertex. Each symbol will be present in the word of its corresponding vertex and neighbors. Therefore, each time the setter answers a guesser's query negatively (which corresponds to a vertex in the graph), the words associated to its associated vertex and its neighborhood are discarded from the possibles words to reveal at the end, and we can interpret this as discarding that vertex and its neighborhood from the graph. As long as there is a non-discarded vertex in the graph, the setter can claim that the encoding of such a vertex is the secret word. This relates the number of nodes, and their respective neighborhoods, that are enough to cover the entire graph (\MDS) and the amount of failed guesses a setter can force (\textsc{Evil Hangman}). The relationship between the two problems is however inverted, the existence of a small dominating set will allow the guesser to quickly discard many options and thus will constrain the victory of our protagonist, the setter. This idea leads therefore to a proof of \conphardness.

\begin{theorem}\label{thm:EvilHangmanIsNPHard}
\textsc{Evil Hangman} is \conphard, even when restricted to languages with words of length 4.
\end{theorem}

\begin{proof}
Given a graph $G = (V,E)$ and a positive integer $d$, let $(G, d)$ be an instance of the \textsc{3-Regular Dominating Set} problem (defined formally below)
\begin{itemize}
\item Theorem~\ref{thm:inddom} proves that \textsc{3-Regular Dominating Set} is \nphard.
\item Lemma~\ref{lemma:encoding} describes how to compute a language $L$ that properly encodes $G$ in polynomial time. The alphabet size $\sigma$ of $L$ corresponds to the number of labels in $V$.
\item Lemma~\ref{lemma:Pos3DomImpliesPosEHangman} proves that if $(G,d)$ is a positive instance of \textsc{3-Regular Dominating Set}, then $(L, d)$ is a negative instance of \textsc{Evil Hangman}; while
\item Lemma~\ref{lemma:PosEHangmanImpliesPos3Dom} proves (by contradiction) that if $(L, d)$ is a positive instance of \textsc{Evil Hangman}, then $(G,d)$ is a negative instance of \textsc{3-Regular Dominating Set}.
\end{itemize}
The combination of Lemmas~\ref{lemma:Pos3DomImpliesPosEHangman} and~\ref{lemma:PosEHangmanImpliesPos3Dom} proves that $(G,d)$ is a positive instance of \textsc{3-Regular Dominating Set} if and only if $(L, d)$ is a negative instance of \textsc{Evil Hangman}, which permits to deduce the \conphardness~of \textsc{Evil Hangman} from the \nphardness~of \textsc{3-Regular Dominating Set} (from Theorem~\ref{thm:inddom})
\end{proof}

In order to ensure that the language only contains words of a fixed length, we consider only the restricted class of graphs where the degree of every node is fixed:

\begin{definition}[k-Regular Graph]
A graph $G=(V,E)$ is  \emph{$k$-regular} if and only if every vertex $v$ in $V$ has exactly $k$ neighbors.
\end{definition}

The main combinatorial problem required for our results is that of 

\begin{definition}[Dominating Vertex Set]
Given a graph $G = (V, E)$, a set of vertices $D \subseteq V$ is \emph{dominating} if and only if every node in $V$ is either a member of $D$ or has a neighbor in $D$.
\end{definition}

The problem we will reduce from is a restricted version of \MDS.

\csproblem{3-Regular Dominating Set}{ Given a $3$-regular graph $G$ and an integer $d$, decide whether the size $\gamma(G)$ of the minimum dominating set is at most $d$.}

The problem of \MDS\ has been intensively studied since the 1970s, and its \nphardness~is known for several classes of graphs (Planar, Perfect, Bipartite, Chordal, Split, etc) \cite{Corneil1991}. Our reduction is based on a result by Kikuno et al.~\cite{kikuno1980np}, that shows $\mathrm{NP}$-completeness for $3$-regular planar graphs. This stronger result implies of course hardness for the broader class of $3$-regular graphs, which is essential to our reduction. 

\begin{theorem}[Kikuno et al.~\cite{kikuno1980np}]\label{thm:inddom} 
	\textsc{3-Regular Dominating Set} is \nphard
\end{theorem} 
	
In order to reduce from this problem, we start by showing a proper way to encode a $3$-regular graph as a language. 

\begin{definition}[Vertex Encoding]
Let $G = (V,E)$ be a $3$-regular graph, where every vertex of $V$ is labeled with a symbol from an alphabet~$\Sigma$. We say that a word of $\Sigma^4$ is a \emph{vertex encoding} of a node $v \in V$ if its first symbol is the label of $v$ followed by the labels of its three neighbors.
\end{definition}

\begin{figure}
\center
	\begin{subfigure}[b]{0.3\textwidth}
		\centering
		\begin{tikzpicture}
			\node[draw, circle, fill=black] (c) at (0,0) {};
			\node[] (lc) at (-0.3, -0.3){$c$};
			
			\node[draw, circle, fill=black] (d) at (2,0){};
			\node[] (ld) at (2.3, -0.3){$d$};
			
			\node[draw, circle, fill=black] (a) at (0,2){};
			\node[] (la) at (-0.3, 2.3){$a$};

			\node[draw, circle, fill=black] (b) at (2,2){};
			\node[] (lb) at (2.3, 2.3){$b$};
			
			\path [-,line width=0.6pt] (a) edge node {} (b);
			\path [-,line width=0.6pt] (a) edge node {} (c);
			\path [-,line width=0.6pt] (a) edge node {} (d);
			\path [-,line width=0.6pt] (b) edge node {} (c);
			\path [-,line width=0.6pt] (b) edge node {} (d);
			\path [-,line width=0.6pt] (c) edge node {} (d);
		\end{tikzpicture}
		\caption{The graph $G$ to encode}
		\label{fig:encodingA}
	\end{subfigure}
	\begin{subfigure}[b]{0.25\textwidth}
		\centering
		\begin{enumerate}
			\item $w_a = \texttt{abcd}$
			\item $w_b = \texttt{bacd}$
			\item $w_c = \texttt{cabd}$
			\item $w_d = \texttt{dabc}$
		\end{enumerate}
		\vspace{1.5em}
		\caption{An encoding of $G$ }
		\label{fig:encodingB}
	\end{subfigure}
	\begin{subfigure}[b]{0.3\textwidth}
		\centering
		\begin{enumerate}
			\item $w_a = \texttt{abcd}$
			\item $w_b = \texttt{badc}$
			\item $w_c = \texttt{cdba}$
			\item $w_d = \texttt{dcab}$
		\end{enumerate}
		\vspace{1.5em}
		\caption{A proper encoding of $G$}
		\label{fig:encodingC}
	\end{subfigure}
\caption{Example of encodings for $K_4$}
\label{fig:encoding}
\end{figure}
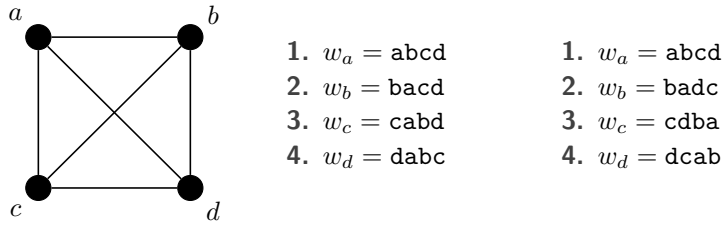

Now, by putting together an encoding of every vertex of a graph, we get a graph encoding as a language.

\begin{definition}[Language Encoding a Graph]
Given a graph $G = (V,E)$ whose vertices are labeled with symbols of an alphabet~$\Sigma$, we say a language $L \subseteq \Sigma^4$ \emph{encodes} $G$ if $L = \{w_1, w_2, \ldots, w_{|V|} \}$ where $w_i$ is a vertex encoding of the $i$-th node.
\end{definition}

An example of such an encoding is presented in Figure~\ref{fig:encodingB}. 
Because vertex encodings can have the neighbors of the represented vertex in any order, there are $(3!)^{|V|}$ possible language encodings for a given $3$-regular graph $G=(V,E)$, and they present different combinatorial properties.
We describe a deterministic way to encode input graphs that permits to identify any word just by knowing of a single letter in it, so that no two words can have the same symbol on the same position. This property greatly simplifies the proof of the reduction in Lemma~\ref{lemma:PosEHangmanImpliesPos3Dom}.

\begin{definition}[Proper Graph Encoding]
	An encoding $L$ of a graph $G = (V, E)$ is said to be \emph{proper} if for every vertex $v \in V$, and every position $p$ in $\{1,2,3,4\}$, there is exactly one word in which the label of vertex $v$ appears in the $p^{th}$ position.
\end{definition}

In Figure~\ref{fig:encodingC} we present an example of a proper encoding. We now prove a key lemma in our reduction: the fact that we can compute a proper encoding of any $3$-regular graph in time polynomial  in the number of vertices of $G$.

\begin{lemma}\label{lemma:encoding}
	Every $3$-regular graph $G$ admits a proper encoding, and such an encoding can be computed in polynomial time.
\end{lemma}

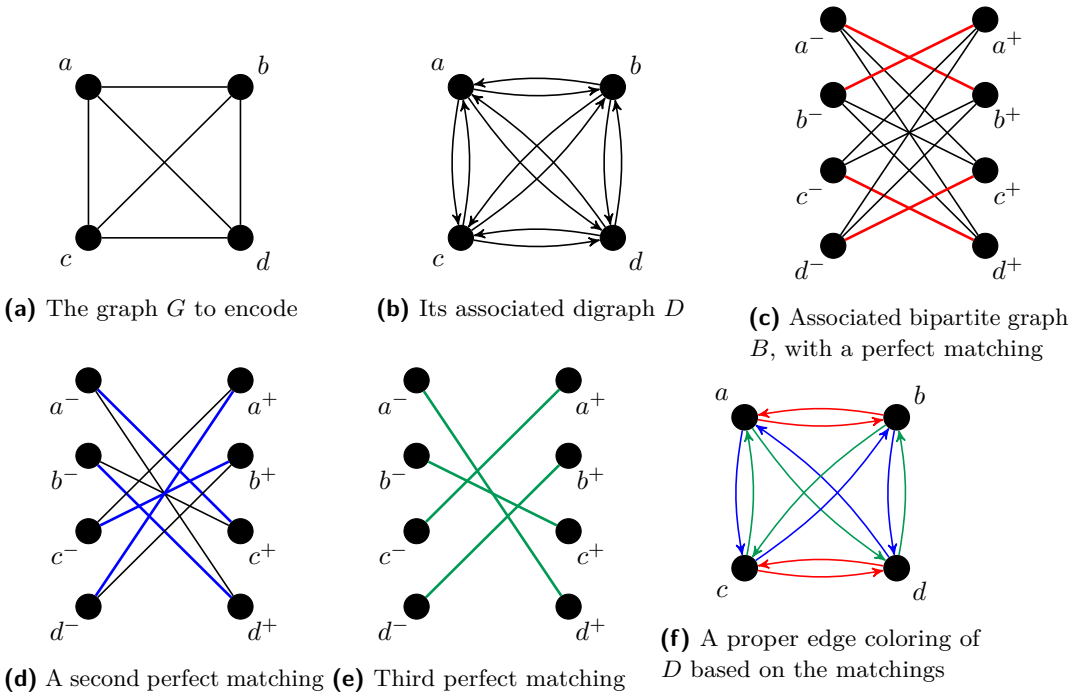
\begin{figure}
	\begin{subfigure}{0.3\textwidth}
		\centering
		\begin{tikzpicture}
			\node[draw, circle, fill=black] (c) at (0,0) {};
			\node[] (lc) at (-0.3, -0.3){$c$};
			
			\node[draw, circle, fill=black] (d) at (2,0){};
			\node[] (ld) at (2.3, -0.3){$d$};
			
			\node[draw, circle, fill=black] (a) at (0,2){};
			\node[] (la) at (-0.3, 2.3){$a$};

			\node[draw, circle, fill=black] (b) at (2,2){};
			\node[] (lb) at (2.3, 2.3){$b$};
			
			\path [-,line width=0.6pt] (a) edge node {} (b);
			\path [-,line width=0.6pt] (a) edge node {} (c);
			\path [-,line width=0.6pt] (a) edge node {} (d);
			\path [-,line width=0.6pt] (b) edge node {} (c);
			\path [-,line width=0.6pt] (b) edge node {} (d);
			\path [-,line width=0.6pt] (c) edge node {} (d);

		\end{tikzpicture}
		\caption{The graph $G$ to encode}
	\end{subfigure}
	\begin{subfigure}{0.3\textwidth}
		\centering
		\begin{tikzpicture}[>=stealth']
			\node[draw, circle, fill=black] (c) at (0,0) {};
			\node[] (lc) at (-0.3, -0.3){$c$};
			
			\node[draw, circle, fill=black] (d) at (2,0){};
			\node[] (ld) at (2.3, -0.3){$d$};
			
			\node[draw, circle, fill=black] (a) at (0,2){};
			\node[] (la) at (-0.3, 2.3){$a$};

			\node[draw, circle, fill=black] (b) at (2,2){};
			\node[] (lb) at (2.3, 2.3){$b$};
			
			\path [->, line width=0.6pt] (a) edge[bend right=10] node {} (b);
			\path [<-, line width=0.6pt] (a) edge[bend left=10] node {} (b);
			\path [->, line width=0.6pt] (a) edge[bend right=10] node {} (c);
			\path [<-, line width=0.6pt] (a) edge[bend left=10] node {} (c);
			\path [->, line width=0.6pt] (a) edge[bend right=10] node {} (d);
			\path [<-,line width=0.6pt] (a) edge[bend left=10] node {} (d);
			\path [->,line width=0.6pt] (b) edge[bend right=10] node {} (c);
			\path [<-,line width=0.6pt] (b) edge[bend left=10] node {} (c);
			\path [->,line width=0.6pt] (b) edge[bend right=10] node {} (d);
			\path [<-,line width=0.6pt] (b) edge[bend left=10] node {} (d);
			\path [->,line width=0.6pt] (c) edge[bend right=10] node {} (d);
			\path [<-,line width=0.6pt] (c) edge[bend left=10] node {} (d);
		\end{tikzpicture}
		\caption{Its associated digraph $D$}
	\end{subfigure}
	\begin{subfigure}{0.3\textwidth}
		\centering
		\begin{tikzpicture}
			\node[draw, circle, fill=black] (dn) at (0,0) {};
			\node[] (ld) at (-0.3, -0.3){$d^{-}$};
			
			\node[draw, circle, fill=black] (cn) at (0,1){};
			\node[] (lc) at (-0.3, 0.7){$c^{-}$};
			
			\node[draw, circle, fill=black] (bn) at (0,2){};
			\node[] (lb) at (-0.3, 1.7){$b^{-}$};

			\node[draw, circle, fill=black] (an) at (0,3){};
			\node[] (la) at (-0.3, 2.7){$a^{-}$};
			
			\node[draw, circle, fill=black] (dp) at (2,0) {};
			\node[] (ld) at (2.3, -0.3){$d^{+}$};
			
			\node[draw, circle, fill=black] (cp) at (2,1){};
			\node[] (lc) at (2.3, 0.7){$c^{+}$};
			
			\node[draw, circle, fill=black] (bp) at (2,2){};
			\node[] (lb) at (2.3, 1.7){$b^{+}$};

			\node[draw, circle, fill=black] (ap) at (2,3){};
			\node[] (la) at (2.3, 2.7){$a^{+}$};
			
			\path [-, red, line width=1.0pt] (an) edge node {} (bp);
			\path [-,line width=0.6pt] (an) edge node {} (cp);
			\path [-,line width=0.6pt] (an) edge node {} (dp);
			
			\path [-, red, line width=1.0pt] (bn) edge node {} (ap);
			\path [-,line width=0.6pt] (bn) edge node {} (cp);
			\path [-,line width=0.6pt] (bn) edge node {} (dp);
			
			\path [-,line width=0.6pt] (cn) edge node {} (ap);
			\path [-,line width=0.6pt] (cn) edge node {} (bp);
			\path [-, red, line width=1.0pt] (cn) edge node {} (dp);
			
			\path [-,line width=0.6pt] (dn) edge node {} (ap);
			\path [-,line width=0.6pt] (dn) edge node {} (bp);
			\path [-, red, line width=1.0pt] (dn) edge node {} (cp);
		\end{tikzpicture}
		\caption{Associated bipartite graph $B$, with a perfect matching}
	\end{subfigure}
	\begin{subfigure}{0.3\textwidth}
		\centering
		\begin{tikzpicture}
			\node[draw, circle, fill=black] (dn) at (0,0) {};
			\node[] (ld) at (-0.3, -0.3){$d^{-}$};
			
			\node[draw, circle, fill=black] (cn) at (0,1){};
			\node[] (lc) at (-0.3, 0.7){$c^{-}$};
			
			\node[draw, circle, fill=black] (bn) at (0,2){};
			\node[] (lb) at (-0.3, 1.7){$b^{-}$};

			\node[draw, circle, fill=black] (an) at (0,3){};
			\node[] (la) at (-0.3, 2.7){$a^{-}$};
			
			\node[draw, circle, fill=black] (dp) at (2,0) {};
			\node[] (ld) at (2.3, -0.3){$d^{+}$};
			
			\node[draw, circle, fill=black] (cp) at (2,1){};
			\node[] (lc) at (2.3, 0.7){$c^{+}$};
			
			\node[draw, circle, fill=black] (bp) at (2,2){};
			\node[] (lb) at (2.3, 1.7){$b^{+}$};

			\node[draw, circle, fill=black] (ap) at (2,3){};
			\node[] (la) at (2.3, 2.7){$a^{+}$};
			
			\path [-, blue, line width=1.0pt] (an) edge node {} (cp);
			\path [-,line width=0.6pt] (an) edge node {} (dp);
			
			\path [-,line width=0.6pt] (bn) edge node {} (cp);
			\path [-, blue, line width=1.0pt] (bn) edge node {} (dp);
			
			\path [-,line width=0.6pt] (cn) edge node {} (ap);
			\path [-, blue, line width=1.0pt] (cn) edge node {} (bp);
			
			\path [-, blue, line width=1.0pt] (dn) edge node {} (ap);
			\path [-,line width=0.6pt] (dn) edge node {} (bp);
		\end{tikzpicture}
		\caption{A second perfect matching}
	\end{subfigure}
		\begin{subfigure}{0.3\textwidth}
		\centering
		\begin{tikzpicture}
			\node[draw, circle, fill=black] (dn) at (0,0) {};
			\node[] (ld) at (-0.3, -0.3){$d^{-}$};
			
			\node[draw, circle, fill=black] (cn) at (0,1){};
			\node[] (lc) at (-0.3, 0.7){$c^{-}$};
			
			\node[draw, circle, fill=black] (bn) at (0,2){};
			\node[] (lb) at (-0.3, 1.7){$b^{-}$};

			\node[draw, circle, fill=black] (an) at (0,3){};
			\node[] (la) at (-0.3, 2.7){$a^{-}$};
			
			\node[draw, circle, fill=black] (dp) at (2,0) {};
			\node[] (ld) at (2.3, -0.3){$d^{+}$};
			
			\node[draw, circle, fill=black] (cp) at (2,1){};
			\node[] (lc) at (2.3, 0.7){$c^{+}$};
			
			\node[draw, circle, fill=black] (bp) at (2,2){};
			\node[] (lb) at (2.3, 1.7){$b^{+}$};

			\node[draw, circle, fill=black] (ap) at (2,3){};
			\node[] (la) at (2.3, 2.7){$a^{+}$};
			
			\path [-, ForestGreen, line width=1.0pt] (an) edge node {} (dp);
			
			\path [-, ForestGreen, line width=1.0pt] (bn) edge node {} (cp);
			
			\path [-, ForestGreen, line width=1.0pt] (cn) edge node {} (ap);
			
			\path [-, ForestGreen, line width=1.0pt] (dn) edge node {} (bp);
		\end{tikzpicture}
		\caption{Third perfect matching}
	\end{subfigure}
	\begin{subfigure}{0.3\textwidth}
		\centering
		\begin{tikzpicture}[>=stealth', every path/.style={line width=0.6pt}]
			\node[draw, circle, fill=black] (c) at (0,0) {};
			\node[] (lc) at (-0.3, -0.3){$c$};
			
			\node[draw, circle, fill=black] (d) at (2,0){};
			\node[] (ld) at (2.3, -0.3){$d$};
			
			\node[draw, circle, fill=black] (a) at (0,2){};
			\node[] (la) at (-0.3, 2.3){$a$};

			\node[draw, circle, fill=black] (b) at (2,2){};
			\node[] (lb) at (2.3, 2.3){$b$};
			
			\path [->, red] (a) edge[bend right=10] node {} (b);
			\path [<-, red] (a) edge[bend left=10] node {} (b);
			\path [->, blue] (a) edge[bend right=10] node {} (c);
			\path [<-, ForestGreen] (a) edge[bend left=10] node {} (c);
			\path [->, ForestGreen] (a) edge[bend right=10] node {} (d);
			\path [<-, blue] (a) edge[bend left=10] node {} (d);
			\path [->, ForestGreen] (b) edge[bend right=10] node {} (c);
			\path [<-, blue] (b) edge[bend left=10] node {} (c);
			\path [->, blue] (b) edge[bend right=10] node {} (d);
			\path [<-, ForestGreen] (b) edge[bend left=10] node {} (d);
			\path [->,red] (c) edge[bend right=10] node {} (d);
			\path [<-,red] (c) edge[bend left=10] node {} (d);
		\end{tikzpicture}
		\caption{A proper edge coloring of $D$ based on the matchings}
		\label{fig:proofencodingF}
	\end{subfigure}
\caption{Illustration for the proof of Lemma~\ref{lemma:encoding} on $K_4$. Note that the encoding resulting of subfigure (f) corresponds to the one presented in Figure~\ref{fig:encodingC}.}
\label{fig:proofencoding}
\end{figure}

\begin{proof}

Let $G = (V,E)$ be a $3$-regular graph, we start by considering the digraph $D = (V', E')$ associated to $G$ where $V' = V$ and $E'$ contains the pairs $(u,v)$ and $(v,u)$ if there was an edge between nodes $u$ and $v$ in $G$. We claim that if there exists a way to color edges in $D$ with $\{\texttt{red, blue, green}\}$ such that every vertex has (i) incoming edges of each different color, and (ii) outgoing edges of each different color, then we can produce a proper encoding based on that. Here's how to do it: if vertex $u$ has a red outgoing edge to $v$, a green outgoing edge to $w$ and a blue outgoing edge to $x$, then we can encode it as $\texttt{uvwx}$. Note that the color of an  edge $u \to v$ determines in which position is $v$ going to appear in the encoding of $u$, and therefore condition (i) over $v$ ensures that the label of $v$ appears in every position, while condition (ii) over $u$ ensures that no more than one vertex is assigned position $p$ on the encoding of $u$. 

In order to find such an edge coloring, we create the undirected bipartite graph $B = (V'', E'')$, where for every vertex $v \in V$, we put two vertices $v^{+}$ and $v^{-}$ in $V''$, and for every edge~$(u,v)$ in $E'$ we put the edges $(u^{+}, v^{-})$ and $(u^{-}, v^{+})$. The partition of $B$ is then, of course, the set of vertices $(\cdot)^{+}$ and the set of vertices $(\cdot)^{-}$. Note that $B$ is also a $3$-regular graph, as every vertex $v$ with neighbors $u, w$ and $x$ in the original graph, it is associated vertex $u^{+}$ is connected with $v^{-}, w^{-}$ and $x^{-}$ in $B$, and $u^{-}$ will be connected to $u^{+}, w^{+}$ and~$x^{+}$.

As a direct consequence of Hall's Marriage Theorem \cite{Hall1935}, every regular bipartite graph has a perfect matching.  Such a perfect matching can be computed in polynomial time using for example the Hopcroft-Karp algorithm \cite{Hopcroft1973}. Let $M$ be the set of edges of a perfect matching computed that way. We can color every edge in $M$ with red. Now, if we remove from $E''$ all the edges of $M$, the bipartite graph is 2-regular, as each node has lost exactly one neighbor. By using Hall's theorem again, we can get a new perfect matching $M'$, whose edges we color with blue. If we now remove all the edges of $M''$, we get a $1$-regular graph, which is itself a perfect matching, and whose edges we color with green. This is enough to get the required coloring in the graph $D$, just by coloring every edge~$(u,v)$ with the same color of the edge~$(u^{-}, v^{+})$.
\end{proof}

The final step of the proof of Theorem~\ref{thm:EvilHangmanIsNPHard} consists in proving that $(G,d)$ is a positive instance of \textsc{3-Regular Dominating Set} if and only if $(L, d)$ is a negative instance of \textsc{Evil Hangman}. Lemma~\ref{lemma:Pos3DomImpliesPosEHangman} proves the forward direction of the statement:

\begin{lemma} \label{lemma:Pos3DomImpliesPosEHangman}
Let $G$ be a $3$-regular graph, and $L$ be the language built from $G$ as described in Lemma~\ref{lemma:encoding}, and let $d$ be an arbitrary integer.  If $(G,d)$ is a positive instance of \textsc{3-Regular Dominating Set}, then $(L, d)$ is a negative instance of \textsc{Evil Hangman}.
\end{lemma}

\begin{proof}
Let's assume $(G,d)$ is a positive instance of the  \textsc{3-Regular Dominating Set} problem. This means that there is a dominating set for $G$ of size at most $d$. Let $\Gamma = = \{ v_1, v_2, \ldots, v_{\gamma(G)} \}$ be such a dominating set. Then, if the guesser makes the sequence of queries $\ell(v_1), \ell(v_2), \ldots, \ell(v_{\gamma(G)})$, where $\ell(v)$ corresponds to the label of vertex $v$, the setter is forced to answer positively at least one of those queries, as otherwise there would be no possible word for her to reveal at the end. Thus far, the guesser has made at most $d-1$ failed queries. As the encoding of the graph $G$ into the language $L$ is proper, a single guess answered positively is enough to uniquely determine the secret word, and therefore the guesser can win the game without making any more failed guesses, implying that the instance $(L, d)$ is negative for \textsc{Evil Hangman}. 
\end{proof}

Lemma~\ref{lemma:PosEHangmanImpliesPos3Dom} proves (by contradition) the reverse direction:

\begin{lemma} \label{lemma:PosEHangmanImpliesPos3Dom}
Let $G$ be a $3$-regular graph, and $L$ be the language built from $G$ as described in Lemma~\ref{lemma:encoding}, and let $d$ be an arbitrary integer.  If $(L, d)$ is a negative instance of \textsc{Evil Hangman}, then  $(G, d)$ is a positive instance of \textsc{3-Regular Dominating Set}.
\end{lemma}

\begin{proof}

We will show the contrapositive statement, namely, that if $(G, d)$ is a negative instance of \textsc{3-Regular Dominating Set} then $(L, d)$ is a positive instance of \textsc{Evil Hangman}.  Let's assume that $(G, d)$ is a negative instance, and thus, $d$ vertices are not enough to dominate the graph. This would mean that for any set $D$ of $d$ vertices, there is at least one vertex $v_D$ which is not dominated by $D$. Consider then that the guesser makes a sequence of $d$ queries, whose associated vertices form the set $D'$. Then, rejecting all those $d$ queries and revealing $w(v_{D'})$ as the secret word is a guaranteed strategy for the setter, meaning that $(L, d)$ is a positive instance of \textsc{Evil Hangman}.

\end{proof}

This concludes the proof of the \conphardness~of \textsc{Evil Hangman} when $L$ is a finite language. 
 We explore its computational complexity for more general types of languages in the next section. 

\section{Generalized Languages}
\label{sec:turing-recogn-dict}

A natural generalization of the \textsc{Evil Hangman} problem is to consider its complexity when played over broader classes of languages, such as Regular, Context Free or Turing computable languages over an alphabet $\Sigma$, projected to words of length $k$ by intersection with $\Sigma^k$.
The lower bound of Theorem~\ref{thm:EvilHangmanIsNPHard} (where the language is finite) can be extended to prove hardness for the cases where the language is defined by a Regular Expression (or, respectively, a Context Free Grammar or a Turing Machine) by observing that a dictionary of $n$ words of length $k$ can be encoded in a Regular Expression (or respectively, a Context Free Grammar or a Turing Machine) of size $(nk)^{O(1)}$, and thus we can construct hard instances for such problems by using the same construction used to prove Theorem~\ref{thm:EvilHangmanIsNPHard}.

In this section we give a result for classes of languages whose associated machines  are strong enough to simulate other machines within their class. Namely, that when the game is played over languages of such classes, the \textsc{Evil Hangman} decision problem is at least as hard as deciding membership of a word in the language.
This implies undecidability for Turing computable languages (given as Turing Machines, abbreviated as TM) and \textrm{PSPACE}-completeness for Context Sensitive Languages \cite{Linz2011} (given as their equivalent Linear Bounded Automata \cite{Kuroda1964}, abbreviated as LBA). 
The proofs of hardness are thus based on reductions to the membership problem, stated below.

\csproblem{Membership}{Given an encoding of a machine (or language) $C$ belonging to a class $\mathcal{C}$, and a word $w$, decide whether $w \in L(C)$}

We consider first a restricted class of languages among Turing Computable languages:

\begin{definition}[Universal Simulation Languages] 	\label{def:universalSimulatiom} 
Let $\mathcal{C}$ be a class of machines, we say that $\mathcal{C}$ allows universal simulation if, given a machine $C$, it is possible to construct in polynomial time a machine $C'$ that accepts exactly the language $\{ \alpha, \beta \}$ if $(C, w) \in \textsc{Membership}(\mathcal{C})$ and $\{ \alpha \}$ otherwise.
\end{definition}

We can now state the key lemma used to prove undecidability of \textsc{Evil Hangman}(TM) and \textrm{PSPACE}-completeness for \textsc{Evil Hangman}(LBA).

\begin{lemma}
Let $\mathcal{C}$ be a class of machines (languages) allowing universal simulation. Then, there is a polynomial time reduction problem from $\textsc{Membership}(\mathcal{C})$ to $\textsc{Evil Hangman}(\mathcal{C})$, that is, \textsc{Evil Hangman} but over a language $L$ defined by an element of $\mathcal{C}$.
	\label{lemma:membershipTohangman}	
\end{lemma}

\begin{proof}
Consider and an arbitrary element $C \in \mathcal{C}$. Because of the universal simulation property, we can construct a machine $C'$ with the behavior specified in Definition~\ref{def:universalSimulatiom}. Now consider $k = 1$ and the instance $(L(C'), d = 1)$. We can see that is a positive instance of \textsc{Evil Hangman} if and only if $C$ accepts on $\beta$, as if the $C$ accepts $\beta$ the dictionary has size $2$ and can force a failed guess, but if $C$ rejects $\beta$, the dictionary will have size $1$ and thus it is not possible to force a failed guess.
\end{proof}

We have now the machinery required to easily prove the following two theorems, that define the computability and complexity of \textsc{Evil Hangman} over Context Sensitive languages and Turing Computable languages.
A reduction  from \textsc{Membership}$(\textsc{PSPACE})$ yields the PSPACE completeness:
\begin{theorem}
\label{thm:hangmanLBAIsPSPACEC}
	\textsc{Evil Hangman} is $\mathrm{PSPACE}$-complete  when the language $L$ is the language defined by an arbitrary Linear Bounded Automaton $M$.
\end{theorem}

\begin{proof}
\textsc{Membership}(LBA) is \textrm{PSPACE}-complete \cite{Karp1972}. This implies membership in \textrm{PSPACE} by a naive simulation of Equation (2). To prove hardness, we can use a result by Feldman et al. \cite{Feldman1973ACO} that states that for every $n$, there is a universal LBA $M_n$ for the class of LBAs using at most $n$ tape symbols, and thus LBAs hold the property of universal simulation (Definition~\ref{def:universalSimulatiom}). As \textsc{Membership}(LBA) is in particular \textrm{PSPACE}-hard, the reduction implies as well hardness for our problem.
\end{proof}

A similar reduction from \textsc{Membership}, but this time from $TM$, yields the undecidability result:

\begin{theorem}
	\label{thm:hangmanTMIsUndecidable}
\textsc{Evil Hangman} is undecidable  when the language $L$ is the language defined by an arbitrary Turing machine $M$.
\end{theorem}

\begin{proof} 
We reduce from $\textsc{Membership}(TM)$, directly from Lemma~\ref{lemma:membershipTohangman}, as Turing Machines trivially hold the property of universal simulation from Definition~\ref{def:universalSimulatiom}.  The fact that membership is undecidable for Turing Machines (Rice's Theorem) concludes the proof.
\end{proof}

This concludes our results about the computational complexity of optimal strategies for the \textsc{Evil Hangman} problem.
In the next section, we summarize our results and outline some remaining open questions.

\section{Discussion}
\label{sec:openquestions}

On one hand, the greedy strategy for \textsc{Evil Hangman} (the one which is commonly given as a programming assignment) can perform arbitrarily bad on certain languages (Theorem~\ref{thm:greedycheaterIsNotCCompetitive}); on the other hand finding an optimal strategy for a given language is \conphard~(Theorem~\ref{thm:EvilHangmanIsNPHard}), and thus we cannot expect a polynomial time algorithm for it unless $\mathrm{P} = \mathrm{NP}$.
Note that the \conphardness~from the setter's perspective implies \nphardness~from the guesser's perspective: Theorem~\ref{thm:EvilHangmanIsNPHard} is equivalent to the \nphardness~of deciding whether a guesser can always win the game without making $d$ failed queries.
Even worse, the optimality of an answer by the evil setter is $\mathrm{PSPACE}$-complete for languages described by Context Sensitive Grammars (Theorem~\ref{thm:hangmanLBAIsPSPACEC}) and undecidable for Turing Computable languages (Theorem~\ref{thm:hangmanTMIsUndecidable}).

Although hard in arbitrary languages, the game of \textsc{Hangman} is traditionally played on natural languages, where alphabets are pretty small, and words are pretty short. Hence it is worth noticing that Equation~\ref{eq:recursive} (given on page~\pageref{eq:recursive}) yields a Fixed Parameter Tractable (\textrm{FPT}) algorithm~\cite{Downey1995,2006-BOOK-ParameterizedComplexityTheory-FlumGrohe} when parameterized over $\ell = \sigma + k$ (size of the alphabet + word length). In particular, it can be implemented in time within $n \cdot 2^{O(\ell)}$, where $n = |L|$.
First, note that the recursive formula goes over all the possible $\sigma^k$ masks, all possible symbols, and all the $2^{\sigma}$ possible subsets of the alphabet. This last term can be immediately optimized by considering only the masks that are present in the dictionary, which are no more than $n2^k$. Note that by considering only those masks, the remaining language (in the subscript of $OPT$ in Equation~\ref{eq:recursive}) is kept implicit.
	Therefore, the total number of cells is bound by $\sigma n2^\sigma 2^k \in n 2^{O(\ell)}$. At every cell we have to choose between at most $2^k$ sub-masks of $M$ and $\sigma$ symbols, and compute $f$ which is done in time within $O(k)$. Thus, the total computational work per cell is within $2^{O(\ell)}$. Multiplying this by the amount of cells gives us the desired result.
It is an open problem whether \textsc{Evil Hangman} becomes \textrm{FPT} when parameterized by $d$, the number of failed guesses allowed. The reduction presented in Lemma \ref{lemma:PosEHangmanImpliesPos3Dom} constitutes an \textrm{FPT} reduction from \MDS\ on $3$-regular graphs. It is well known that \MDS\ on general graphs is complete for the class $W[2]$ (Downey et al.~\cite{Downey1995,2006-BOOK-ParameterizedComplexityTheory-FlumGrohe}). However, $k$-regular graphs are $K_{k+1,k+1}$ free, and thus the result of Telle et al. \cite{Telle2012} implies \MDS\ to be \textrm{FPT} when parameterized by the size of the set. This of course does not imply that \textsc{Evil Hangman} is \textrm{FPT} under such a parameterization: only that we cannot derive fixed parameter intractability from the reduction to dominating set in $3$-regular graphs presented in Lemmas \ref{lemma:Pos3DomImpliesPosEHangman}
and \ref{lemma:PosEHangmanImpliesPos3Dom}.
